\documentclass[11pt,twoside]{article}
\usepackage{a4wide}
\usepackage{amsthm}
\theoremstyle{definition}
       \newtheorem{defn}{Definition}[section]
       \newtheorem{expl}[defn]{Example}
       \newtheorem{assumption}[defn]{Assumption}
       \newtheorem{algorithm}[defn]{Algorithm}
       \newtheorem{conv}[defn]{Convention}
       \newtheorem{rem}[defn]{Remark}
\theoremstyle{plain}
       \newtheorem{thm}[defn]{Theorem}
       \newtheorem{lem}[defn]{Lemma}
       
       \newtheorem{prop}[defn]{Proposition}
       \newtheorem{cor}[defn]{Corollary}

\usepackage[english]{babel}
\usepackage{bbm}
\usepackage{amsthm}
\usepackage{amssymb}
\usepackage{amsmath}
\usepackage{xspace}

\newcommand{\into}{\hookrightarrow}
\newcommand{\cf}{\mathbbm{1}}
\newcommand{\Land}{\bigwedge}

\newcommand{\ID}{\mathit{ID}}
\newcommand{\parr}{\rightharpoonup}

\newcommand{\Lang}{\mathcal{L}}	
	
\newcommand{\FLang}{\mathcal{F}}	
\newcommand{\Struct}{\mathcal{M}}
\newcommand{\pls}{\Lambda}

\newcommand{\gldiamond}[1]{\Diamond_{#1}}

\newcommand{\Sub}{\mathit{MSub}}

\newcommand{\Nat}{{\mathbb{N}}}

\newcommand{\Rat}{{\mathbb{Q}}}

\newcommand{\PDist}{D_\omega}%% Provisionary!!

\newcommand{\Bag}{\mathcal{B}}
\newcommand{\Prop}{\mathsf{Prop}}

\newcommand{\size}{\operatorname{size}}
\newcommand{\maxsize}{\operatorname{maxsize}}
\newcommand{\powerset}{{\mathcal P}}

\newcommand{\Th}{\mathsf{Th}}
\newcommand{\modelsCA}{\models}
\newcommand{\modelsOS}{\models^1}
\newcommand{\modelsPL}{\models^0}

\newcommand{\Sem}[1]{{[\![#1]\!]}}

\newcommand{\PLSem}[1]{{\Sem{#1}^0}}

\newcommand{\conj}{\wedge}
\newcommand{\disj}{\vee}
\newcommand{\modimpl}{\to}
\newcommand{\modiff}{\leftrightarrow}
\newcommand{\impl}{\Rightarrow}

\newcommand{\rank}{\operatorname{rank}}
\newcommand{\mi}[1]{\mathit{#1}}
\newcommand{\Op}{{\mathit{op}}}

\newcommand{\id}{{\mathit{id}}}
\newcommand{\Id}{{\mathit{Id}}}

\newcommand{\Cat}{\mathsf}
\newcommand{\Set}{\Cat{Set}}

\newcommand{\argument}{\_\!\_}%{\underline{\;\;}}

\newcommand{\ModOp}{L}%% Default name for some modal operator 
\newcommand{\probset}{\mathfrak{P}}%%Default name for a set of
\newcommand{\contrapow}{\mathcal{Q}}
\newcommand{\CK}{\mi{CK}}
\newcommand{\CKId}{$\CK\!+\!\ID$\xspace}
\newcommand{\MP}{\mathit{MP}}
\newcommand{\CKmp}{$\CK\!+\!\MP$\xspace}
\newcommand{\ThreeNb}{N_3}
\newcommand{\AgFunct}{\mathcal{A}}
\newcommand{\meet}{\wedge}
\newcommand{\Meet}{\bigwedge}

\newcommand{\CF}{\mathit{Cf}}
\newcommand{\CFid}{\mathit{Cf_{\ID}}}
\newcommand{\CFmp}{\mathit{Cf_{\MP}}}

\newcommand{\NP}{NP\xspace}
\newcommand{\PSPACE}{PSPACE\xspace}
\newcommand{\EXPT}{EXPTIME\xspace}
\newcounter{blubber}

\newenvironment{axarray}%
	       {\begin{array}{@{\hspace{5em}}p{5em}p{50em}}}{\end{array}}
\newenvironment{sparenumerate}
{\begin{list}
  {\arabic{blubber}.}
  {\usecounter{blubber}
   \setlength{\leftmargin}{0pt}
    \setlength{\parsep}{0pt}
    \setlength{\itemindent}{3ex}
    \setlength{\itemsep}{2pt}   
    \setlength{\listparindent}{3ex}
  }
}
{\end{list}}

\newenvironment{algenumerate}
{\begin{list}
  {\arabic{blubber}.}
  {\usecounter{blubber}
   \setlength{\leftmargin}{0ex}
    \setlength{\parsep}{0pt}
    \setlength{\itemindent}{3ex}
    \setlength{\itemsep}{2pt}   
    \setlength{\listparindent}{3ex}
  }
}
{\end{list}}

\newenvironment{sparitemize}
{\begin{list}{$\bullet$}{
		\setlength{\topsep}{0pt}
    \setlength{\leftmargin}{0pt}
    \setlength{\parsep}{0pt}
    \setlength{\itemindent}{4ex}
    \setlength{\itemsep}{0pt}
  }
}{\end{list}}

\newcommand{\eat}[1]{}

\begin{document}

\title{Shallow Models for Non-Iterative Modal Logics}
\author{Lutz Schr{\"o}der
\\
DFKI-Lab Bremen\\
and Department of Computer Science, Universit\"at Bremen
 \and 
Dirk Pattinson\\
Department of Computing\\ Imperial College London}
\maketitle             
\thispagestyle{empty}

\begin{abstract}
  \noindent The methods used to establish $\mi{PSPACE}$-bounds for
  modal logics can roughly be grouped into two classes: syntax driven
  methods establish that exhaustive proof search can be performed in
  polynomial space whereas semantic approaches directly construct
  shallow models. In this paper, we follow the latter approach and
  establish generic $\mi{PSPACE}$-bounds for a large and heterogeneous
  class of modal logics in a coalgebraic framework. In particular, no
  complete axiomatisation of the logic under scrutiny is needed. This
  does not only complement our earlier, syntactic, approach
  conceptually, but also covers a wide variety of new examples which
  are difficult to harness by purely syntactic means. Apart from
  re-proving known complexity bounds for a large variety of
  structurally different logics, we apply our method to obtain
  previously unknown $\mi{PSPACE}$-bounds for Elgesem's logic of
  agency and for graded modal logic over reflexive frames.
\end{abstract}

% \begin{abstract}
%   Establishing an upper complexity bound for a given modal logic can
%   be a non-trivial problem, usually tackled using a wide range of
%   techniques, often in an ad-hoc manner. As domain-specific modal
%   logics, frequently non-normal, abound in the literature and new
%   ones appear at regular intervals, it is therefore desirable to
%   develop a generic framework for deriving such bounds
%   systematically. Here, we present two semantic criteria for
%   non-iterative modal logics (i.e. logics whose axioms do not nest
%   modalities) to be decidable in PSPACE, typically a tight upper
%   bound; the generality of our approach is based on coalgebraic
%   methods. These results both complement an earlier tableau-based
%   method and extend the class of logics covered by generic
%   techniques. In some cases, e.g. conditional logics, the semantic
%   criteria are verified quite easily, and in others, notably various
%   logics of quantitative uncertainty, they follow by dissecting
%   off-the-shelf results. As example applications of our methods, we
%   prove novel PSPACE upper bounds for Elgesem's modal logic of
%   agency and for a generalisation of the graded version of T.
% \end{abstract}

\section{Introduction}

\noindent Special purpose modal logics often combine expressivity and
decidability, usually in a low complexity class. In the absence of
fixed point operators, these logics are frequently decidable in
\PSPACE, i.e.\ not dramatically worse than propositional logic.
While lower \PSPACE bounds for modal logics can typically be
obtained directly from seminal results of Ladner~\cite{Ladner77} by
embedding a \PSPACE-hard logic such as $K$ or $KD$, upper bounds
are often non-trivial to establish. In particular \PSPACE upper
bounds for non-normal logics have recently received much attention:
\begin{sparitemize}
\item A \PSPACE upper bound for graded modal
  logic~\cite{Fine72} is obtained using a constraint set algorithm
  in~\cite{Tobies01}. This corrects a previously published incorrect
  algorithm and refutes a previous \EXPT hardness conjecture.
\item More recently, a $\PSPACE$ upper bound for Presburger modal
  logic (which contains graded modal logic and majority
  logic~\cite{PacuitSalame04}) has been established using a
  Ladner-type algorithm~\cite{DemriLugiez06} .
\item Using a variant of a shallow neighbourhood frame construction
  from~\cite{Vardi89}, a \PSPACE upper bound for coalition logic
  is established in~\cite{Pauly02}.
\item \PSPACE upper bounds for $\CK$ and related conditional
  logics~\cite{Chellas80} are obtained
  in~\cite{OlivettiEA07} by a detailed proof-theoretic
  analysis of a labelled sequent calculus.
\end{sparitemize}

The methods used to obtain these results can be broadly grouped into
two classes. Syntactic approaches presuppose a complete tableaux or
Gentzen system and establish that proof search can be performed in
polynomial space. Semantics-driven approaches, on the other hand,
directly construct shallow tree models. Both approaches are intimately
connected in the case of normal modal logics interpreted over Kripke
frames: counter models can usually be derived directly from search
trees~\cite{HalpernMoses92}. It should be noted that this method is
not immediately applicable in the non-normal case, where the structure
of models often goes far beyond mere graphs.

Using coalgebraic techniques, we have previously shown
\cite{SchroderPattinson06} that the syntactic approach uniformly
generalises to a large class of modal logics: starting from a
\emph{one-step complete} axiomatisation, we have applied
\emph{resolution closure} to obtain complete tableaux systems. Generic
\PSPACE-bounds follow if the ensuing rule set is
\emph{\PSPACE-tractable}. Here, we present a different, semantic, set
of methods to establish uniform \PSPACE bounds by directly
constructing shallow models for logics subject to the \emph{one-step
  polysize model property}, or a variant of the latter. In particular,
no axiomatisation of the logic itself is needed.

Apart from the fact that both methods use substantially different
techniques, they apply to different classes of examples. While it is
e.g.\ relatively easy to obtain a resolution closed rule set for
coalition logic \cite{Pauly02}, proving the one-step polysize model
property for (the coalgebraic semantics of) coalition logic is a
non-trivial task. On the other hand, small one-step models are
comparatively easy to construct for complex modal logics such as
probabilistic modal logic~\cite{FaginHalpern94} or Presburger modal
logic~\cite{DemriLugiez06} that are not straightforwardly amenable to
the syntactic approach via resolution closure, either because no
axiomatisation has been given or because the complexity of the
axiomatisation makes the resolution closure hard to harness.

Moreover, the present semantic approach to \PSPACE-bounds takes a
significant step to overcome an important barrier in the coalgebraic
treatment of modal logics.  Existing decidability and completeness
results~\cite{Pattinson03,CirsteaPattinson04,Schroder06,SchroderPattinson06}
are limited to \emph{rank-1 logics}, given by axioms whose modal
nesting depth is uniformly equal to one. While this already
encompasses a large class of examples (including all logics mentioned
so far),
%graded~\cite{Fine72} and
%probabilistic~\cite{HeifetzMongin01} modal logic, majority logic and
%Pressburger modal logic \cite{DemriLugiez06}, the modal logic of
%probability \cite{FaginHalpern94}, the conditional logics $\CK$ and
%\CKId \cite{Chellas80} as well as various deontic
%logics\cite{Goble04}), 
the semantic model construction in the present paper applies to
\emph{non-iterative logics}~\cite{Lewis74}, i.e.\ logics axiomatised
without nested modalities (rank-1 logics additionally exclude
top-level propositional variables). Despite the seemingly minute
difference between the two classes of logics, this generalisation is
not only technically non-trivial but also substantially extends the
scope of the coalgebraic method. Besides the modal logic $T$, the
class of non-iterative logics includes e.g.\ all conditional logics
covered in~\cite{OlivettiEA07} (of which only 4 are rank-1), in
particular \CKmp~\cite{Chellas80}, as well as Elgesem's logic of
agency~\cite{Elgesem97,GovernatoriRotolo05} and the graded version
$Tn$ of $T$~\cite{Fine72}.

As in \cite{SchroderPattinson06}, we work in the framework of
\emph{coalgebraic modal logic}~\cite{Pattinson03} to obtain results
that are parametric in the underlying semantics of particular logics.
While normal modal logics are usually interpreted over Kripke frames,
non-normal logics see a large variety of different semantics, e.g.\
probabilistic systems~\cite{FaginHalpern94}, frames with ordered
branching~\cite{DemriLugiez06}, game frames~\cite{Pauly02}, or
conditional frames~\cite{Chellas80}. The coalgebraic treatment allows
us to encapsulate the semantics in the choice of a \emph{signature
  functor}, whose coalgebras then play the role of models, leading to
results that are unformly applicable to a large class of different
logics.

Since the class of \emph{all} coalgebras for a given signature functor
can always be completely axiomatised in rank~1~\cite{Schroder06}, in
analogy to the fact that the $K$-axioms are complete for the class of
\emph{all} Kripke frames, the standard coalgebraic approach is not
directly applicable to non-iterative logics. To overcome this
limitation, we introduce the new concept of interpreting modal logics
over coalgebras for \emph{copointed functors}, i.e.\ functors $T$
equipped with a natural transformation of type $T \to \Id$.

In this setting, our main technical tool is to cut back model
constructions from modal logics to the level of \emph{one-step logics}
which semantically do not involve state transitions, and then
amalgamate the corresponding \emph{one-step models} into shallow
models for the full modal logic, which ideally can be traversed in
polynomial space.  For this approach to work, the logic at hand needs
to support a small model property for its one-step fragment, the
\emph{one-step polysize model property (OSPMP)}. Our first main
theorem shows that the OSPMP guarantees decidability in polynomial
space.  Crucially, the OSPMP is much easier to establish than a
shallow model property for the logic itself.  To reprove e.g.\
Ladner's \PSPACE upper bound for $K$, one just observes that to
construct a set that intersects $n$ given sets, one needs at most $n$
elements.  For the conditional logics $\CK$, \CKId, and \CKmp, the
OSPMP is similarly easy to check.  For other logics, in particular
various logics of quantitative uncertainty, the OSPMP can be obtained
by sharpening known off-the-shelf results. As a new result, we
establish the OSPMP for Elgesem's logic of agency to obtain a
previously unknown \PSPACE upper bound.

As a by-product of our construction, we obtain \NP-bounds for
the bounded rank fragments of all logics with the OSPMP, generalizing
the corresponding result for the logics $K$ and $T$
from~\cite{Halpern95} to a large variety of structurally different
(non-iterative) logics.

While the OSPMP is usually easy to establish, a weaker property, the
\emph{one-step pointwise polysize model property (OSPPMP)}, can be
used in cases where the OSPMP fails, provided that the signature
functor supports a notion of pointwise smallness for overall
exponential-sized one-step models. This allows traversing
exponentially branching shallow models in polynomial space by dealing
with the successor structures of single states in a pointwise fashion.
Our second main result, which yields \PSPACE upper bounds for logics
with the OSPPMP, is applied to reprove the known \PSPACE bound for
Presburger modal logic~\cite{DemriLugiez06} and to derive a new
\PSPACE bound for Presburger modal logic over reflexive frames, and
hence for $Tn$~\cite{Fine72} (which was so far only known to be
decidable~\cite{FattorosiBarnabaDeCaro85}). The latter result extends
straightforwardly to a description logic with role hierarchies,
qualified number restrictions, and reflexive roles.

\section{Coalgebraic Modal Logic}\label{sec:prelim}

\noindent We recall the coalgebraic interpretation of modal logic and
extend it to non-iterative logics using copointed functors.

A \emph{modal signature} $\Lambda$ is a set of modal operators with
associated finite arity. The signature $\Lambda$ determines two
languages: firstly, the \emph{one-step logic} of $\Lambda$, whose
formulas $\psi,\dots$ (the \emph{one-step formulas}) over a set $V$ of
propositional variables are defined by the grammar
\begin{equation*}
\psi :: = \bot\ \mid \psi_1\conj\psi_2\ \mid\ \neg\psi\ \mid\ \ModOp
(\phi_1,\dots,\phi_n),
\end{equation*}
where $\ModOp\in\Lambda$ is $n$-ary and the $\phi_i$ are propositional
formulas over $V$; and secondly, the \emph{modal logic} of $\Lambda$, whose set
$\FLang(\Lambda)$ of \emph{$\Lambda$-formulas} $\psi,\dots$ is
defined by the grammar
\begin{equation*}
\psi :: =  \bot\ \mid\psi_1\conj\psi_2\ \mid\ \neg\psi\ \mid\ \ModOp
(\psi_1,\dots,\psi_n).
\end{equation*}
Thus, the modal logic of $\Lambda$ is distinguished from the one-step
logic in that it admits nested modalities. The boolean operations
$\disj$, $\modimpl$, $\modiff$, $\top$ are defined as usual. The
\emph{rank} $\rank(\phi)$ of $\phi\in\FLang(\Lambda)$ is the maximal
nesting depth of modalities in $\phi$ (note however that the notion of
rank-$1$ logic~\cite{Schroder06,SchroderPattinson06} is stricter than
suggested by this definition, as it excludes top-level propositional
variables in axioms; the latter are allowed only in non-iterative
logics). We denote by $\FLang_n(\pls)$ the set of formulas of rank at
most $n$; we refer to the languages $\FLang_n(\Lambda)$ as
\emph{bounded-rank fragments}.

We treat one-step logics as a technical tool in the study of modal
logics. However, one-step logics also appear as logics of independent
interest in the
literature~\cite{FaginHalpernMegiddo90,HalpernPucella02,HalpernPucella02b}.
One of the central ideas of coalgebraic modal logic is that properties
of the full modal logic, such as soundness, completeness, and
decidability, can be reduced to properties of the much simpler
one-step logic. This is also the spirit of the present work, whose
core is a construction of polynomially branching shallow models for
the modal logic assuming a small model property for the one-step
logic.

The semantics of both the one-step logic and the modal logic of
$\Lambda$ are parametrized coalgebraically by the choice of a set
functor. The standard setup of coalgebraic modal logic using all
coalgebras for a plain set functor covers only \emph{rank-1 logics},
i.e.\ logics axiomatised one-step formulas~\cite{Schroder06} (a
typical example is the $K$-axiom $\Box (a \modimpl b)\modimpl \Box
a\modimpl\Box b$). Here, we improve on this by considering the class
of coalgebras for a given \emph{copointed} set functor, which enables
us to cover the more general class of \emph{non-iterative logics},
axiomatised by arbitrary formulas without nested modalities (such as
the $T$-axiom $\Box a\modimpl a$). We follow a purely semantic
approach and hence do not formally consider axiomatisations in the
present work (where we do mention axioms, this is for solely
explanatory purposes). However, the extended scope of the new
framework and its relation to non-iterative modal logics (which can be
made precise in the same way as for plain functors and rank-1
logics~\cite{Schroder06}) will become clear in the examples.

In general, a copointed functor $(T,\epsilon)$ consists of a functor
$T: \Set \to \Set$, where $\Set$ is the category of sets, and a
natural transformation $\epsilon:T\to\Id$.
%A coalgebra for $T$ is
%then a right inverse for some $\epsilon_X$. 
For our present purposes, a slightly restricted notion is more
convenient:
\begin{defn}\label{def:coalg}
A \emph{(restricted) copointed functor} $S$ with \emph{signature
functor} $S_0:\Set\to\Set$ is a subfunctor of $S_0\times\Id$ (where
$(S_0\times\Id)X=S_0X\times X$). We say that $S$ is \emph{trivially
copointed} if $S=S_0\times\Id$. An \emph{$S$-coalgebra}~$A=(X,\xi)$
consists of a set $X$ of \emph{states} and a \emph{transition} function
$\xi: X \to S_0X$ such that $(\xi(x),x)\in SX$ for all $x$.
%A \emph{$\powerset
%X$-valuation} $\pi$ for a set $V$ of propositional variables assigns to
%each $a\in V$ a set $\pi(a)\subseteq X$ of states satisfying $a$. A
%\emph{$T$-model} $M=(A,\pi)$ over $V$ consists of a $T$-coalgebra
%$A=(X,\xi)$ and a $\powerset X$-valuation $\pi$ for $V$.
\end{defn}

\begin{rem}
  The modal logic $\FLang(\Lambda)$ does not explicitly include
  propositional variables. These may be regarded as nullary modal
  operators in $\Lambda$; their semantics is then defined over
  coalgebras for $S_0\times\powerset(V)$, where $V$ is the set of
  variables (cf.\ also e.g.~\cite{Schroder06}). We omit discussion of
  propositional variables in the examples, even in cases like the
  modal logic of probability that become trivial in the absence of
  variables; our treatment extends straightforwardly to the case with
  variables in the manner just indicated.
  %, or more systematically using
  %the modularity results of~\cite{SchroderPattinson07}.
\end{rem}\noindent
We view coalgebras as generalized transition systems: the transition
function maps a state to a structured set of successors and
observations, with the structure prescribed by the signature
functor. Thus, the latter encapsulates the branching type of the
underlying transition systems. Copointed functors additionally impose
local frame conditions that relate a state to the collection of its
successors.
% The
%counit serves to represent information on the present state: a given
%functor $T:\Set\to\Set$ can be turned into a \emph{trivial} copointed
%functor $T'$ given by $T'X=TX\times X$, with the second projection as
%counit. All subfunctors $S$ of such a $T'$ are also copointed; all our
%examples are of this nature. In this setting, we will typically
%represent an $S$-coalgebra as a $T$-coalgebra $(X,\xi)$, subject to the
%condition that $(\xi(x),x)\in SX$ for all $x\in X$.

\begin{assumption}\label{ass:injective}
We assume w.l.o.g.\ that $S_0$ preserves injective
maps~\cite{Barr93}, and even $S_0X \subseteq S_0Y$ in case $X \subseteq
Y$, and that $S$ is \emph{non-trivial}, i.e.\ $SX=\emptyset\Rightarrow
X=\emptyset$.
\end{assumption}
\noindent Generalising earlier work (e.g.~\cite{Jacobs00,Kurz01}),
\emph{coalgebraic modal logic} abstractly captures the interpretation
of modal operators as polyadic predicate
liftings~\cite{Pattinson03,Schroder05},
\begin{defn}\label{def:lifting}
An \emph{$n$-ary predicate lifting} ($n\in\Nat$) for 
$S_0$ is a natural transformation
\begin{displaymath}
\lambda:\contrapow^n\to \contrapow\circ S_0^\Op,
\end{displaymath}
where~$\contrapow$ denotes the contravariant powerset functor
$\Set^\Op\to\Set$ (i.e.\ $\contrapow X$ is the powerset $\powerset(X)$,
and $\contrapow f (A)=f^{-1}[A]$), and $\contrapow^n$ is defined by
$\contrapow^n X=(\contrapow X)^n$.
\end{defn}\noindent
A coalgebraic semantics for $\Lambda$ is formally defined as a
\emph{$\Lambda$-structure} $\Struct$ (\emph{over $S$}) consisting of a
copointed functor $S$ with signature functor $S_0$ and an assignment
of an $n$-ary predicate lifting $\Sem{\ModOp}$ for $S_0$ to every
$n$-ary modal operator $L\in\Lambda$. When $S$ is trivially copointed,
we will also call $\Struct$ a \emph{simple $\Lambda$-structure (over
  $S_0$)}.  \emph{We fix the notation $\Lambda$, $\Struct$, $S$, $S_0$
  throughout the paper}.  The semantics of the modal language
$\FLang(\Lambda)$ is then given in terms of a satisfaction relation
$\modelsCA_C$ between states $x$ of $S$-coalgebras $C=(X, \xi)$ and
$\FLang(\pls)$-formulas over $V$. The relation $\modelsCA_C$ is
defined inductively, with the usual clauses for boolean operators.
The clause for an $n$-ary modal operator $\ModOp$ is
\begin{equation*}
  x \modelsCA_C \ModOp (\phi_1,\dots,\phi_n)\; \Leftrightarrow\; \xi(x)
  \in \Sem{\ModOp}(\Sem{\phi_1}_C,\dots,\Sem{\phi_n}_C)
\end{equation*}
where $\Sem{\phi}_C=\{ x \in X \mid x \modelsCA_C \phi\}$.  We drop the
subscripts $C$ when clear from the context.  Our main interest is in the
(local) \emph{satisfiability problem} over $\Struct$:
\begin{defn}\label{def:lang}
An $\FLang(\pls)$-formula $\phi$ is \emph{satisfiable}
if there exist an $S$-coalgebra $C$ and a state $x$ in $C$ such that
$x\modelsCA_C\phi$. Dually, $\phi$ is \emph{valid} if
$x\modelsCA_C\phi$ for all $C,x$.
\end{defn}\noindent
Contrastingly, the semantics of the one-step logic is given in terms of
satisfaction relations $\modelsOS_{X,\tau}$ between elements $t\in
S_0X$ and one-step formulas over $V$, where $X$ is a set and $\tau$ is a
\emph{$\powerset(X)$-valuation} for $V$, i.e.\ a map
$\tau:V\to\powerset(X)$. The valuation $\tau$ canonically induces an
interpretation of $\PLSem{\phi}\tau\subseteq X$ of propositional
formulas $\phi$ over $V$. We write $X,\tau\modelsPL\phi$ if
$\PLSem{\phi}\tau=X$. The relation $\modelsOS_{X,\tau}$ is then defined
by the usual clauses for boolean operators, and
\begin{equation*}
t\modelsOS_{X,\tau}\ModOp (\phi_1,\dots,\phi_n) 
\;\Leftrightarrow\;
t\in\Sem{\ModOp}(\PLSem{\phi_1}\tau,\dots,\PLSem{\phi_n}\tau).
\end{equation*}
Note in particular that the semantics of the one-step logic does not
involve a notion of state transition.  
%For a one-step formula $\psi$, we
%put $\OSSem{\psi}\tau=\{t\in S_0X\mid t\modelsOS_{X,\tau}\psi\}$.
\begin{defn}\label{def:osm}
A \emph{one-step model} $(X,\tau,t,x)$ over $V$ consists of a set $X$, a
$\powerset(X)$-valuation $\tau$ for $V$, $t\in S_0X$, and $x\in X$ such
that $(t,x)\in SX$. The latter condition is vacuous if $S$ is trivially
copointed, in which case we omit the mention of $x$. For a one-step
formula $\psi$ over $V$, $(X,\tau,t,x)$ is a \emph{one-step model of
$\psi$} if $t\modelsOS_{X,\tau}\psi$. 
\end{defn}\noindent
\noindent We recall some basic notation:
\begin{defn}\label{def:propstuff}
We denote the set of propositional formulas over a set~$Z$, generated by
the basic connectives~$\neg$ and~$\conj$, by $\Prop(Z)$. We use
variables $\epsilon$ etc.\ to denote either nothing or $\neg$. Thus, a
\emph{literal} over~$Z$ is a formula of the form $\epsilon a$, with
$a\in Z$. A \emph{(conjunctive) clause} is a finite, possibly empty,
disjunction (conjunction) of literals.  We denote by $\Lambda(Z)$ the set
$\{\ModOp (a_1,\dots,a_n)\mid\ModOp\in\Lambda\textrm{
$n$-ary},a_1,\dots,a_n\in Z\}$.  
%Given a set $V$ of propositional
%variables, a \emph{$Z$-substitution} is a map $\sigma:V\to Z$; we denote
%the result of applying $\sigma$ to a formula $\phi$ over $V$ by
%$\phi\sigma$.
\end{defn}\noindent
%The set of all (conjunctive)
%clauses over $Z$ is denoted by $\Clause(Z)$ ($\ConjClause(Z)$). We
In the above notation, the set of one-step formulas over $V$
is $\Prop(\Lambda(\Prop(V)))$.  The one-step logic may alternatively
be presented in terms of pairs of formulas separating out the lower
propositional layer:
\begin{defn}\label{def:onestepmodel}
A \emph{one-step pair} $(\phi,\psi)$ over $V$ consists of formulas
$\psi\in\Prop(\Lambda(V))$ and $\phi\in\Prop(V)$. A one-step model
$(X,\tau,t,x)$ is a \emph{one-step model of $(\phi,\psi)$} if
$X,\tau\modelsPL\phi$ and $t\modelsOS_{X,\tau}\psi$. 
\end{defn}\noindent
In analogy to the equivalence between axioms and one-step rules
described in~\cite{Schroder06}, one-step pairs and one-step formulas may
replace each other for purposes of satisfiability:
\begin{lem}\label{lem:pairs}
For every one-step pair $(\phi,\psi)$ over $V$ with $\phi$ satisfiable,
there exists a $\Prop(V)$-substitution $\sigma$ such that the one-step
formula $\psi\sigma$ is equivalent to $(\phi,\psi)$ in the sense that if
$t\modelsOS_{X,\tau}\psi\sigma$ then $t\modelsOS_{X,\sigma\tau}(\phi,\psi)$,
and if $t\modelsOS_{X,\tau}(\phi,\psi)$ then $t\modelsOS_{X,\tau}\psi\sigma$
(and $\sigma\tau=\tau$).  Here, $\sigma\tau$ denotes the
$\powerset(X)$-valuation taking $a$ to $\PLSem{\sigma(a)}\tau$.

Conversely, we have, for $\psi\in\Prop(\Lambda(\Prop(V)))$, an
equivalent one-step pair $(\phi,\psi_1)$ over $V\cup W$, where $\psi$
decomposes as $\psi\equiv\psi_1\sigma$, with
$\psi_1\in\Prop(\Lambda(W))$, $\sigma$ a $\Prop(V)$-substitution,
and $V\cap W=\emptyset$, and where $\phi$ is the conjunction of the
formulas $a\modiff\sigma(a)$, $a\in W$. Here, restricting valuations to
$V$ induces a bijection between one-step models of $(\phi,\psi_1)$ and
one-step models of $\psi$.
\end{lem}\noindent
\begin{proof}

The second claim is clear. The first claim is proved as follows.  As
in~\cite{Schroder06}, let $\kappa$ be a satisfying truth valuation for
$\phi$ and put $\sigma(a)=a\conj\phi$ if $\kappa(a)=\bot$, and
$\sigma(a)=\phi\modimpl a$ otherwise. Then $\phi\sigma$ and the formulas
$\phi\modimpl(a\modiff\sigma(a))$, for $a\in V$, are
tautologies~\cite{Schroder06}. Both directions of the first claim now
follow straightforwardly. 
\end{proof}
%In this case,
%$(\phi,\psi)$ is \emph{(one-step) satisfiable}.
\noindent The coalgebraic approach subsumes many interesting modal
logics, including e.g.\ graded and probabilistic modal logics and
coalition logic~\cite{SchroderPattinson06}. Below, we present the most
basic examples, the modal logics $K$ and $T$, as well as various
conditional logics and logics of quantitative uncertainty. The
treatment of Elgesem's modal logic of agency is deferred to
Sect.~\ref{sec:elgesem}.
\begin{expl}\label{expl:logics}
\begin{sparenumerate}
\item \emph{Modal logic $K$:} \label{expl:Kripke} Let
  $\Lambda=\{\Box\}$, with $\Box$ a unary modal operator. We define a
  simple $\Lambda$-structure over the covariant powerset functor
  $\powerset$ (i.e.\ $\powerset X$ is powerset, and $\powerset
  f(A)=f[A]$) by putting $\Sem{\Box}_X(A)=\{B\in \powerset X\mid
  B\subseteq A\}$. Naturality of $\Sem{\Box}$ is just the equivalence
  $f[B]\subseteq A \iff B\subseteq f^{-1}[A]$.

  $\powerset$-coalgebras are Kripke frames, and $\powerset$-models are
  Kripke models. The modal logic of $\Lambda$ is precisely the modal
  logic $K$, equipped with its standard Kripke semantics. Contrastingly,
  a one-step formula over $V$ is a propositional combination of atoms of
  the form $\Box\phi$, where $\phi\in\Prop(V)$. For $A\in\powerset X$,
  we have $A\modelsOS_{X,\tau}\Box\phi$ iff
  $A\subseteq\PLSem{\phi}\tau$. One easily checks that the one-step
  logic is \NP-complete, while the modal logic $K$ is
  \PSPACE-complete~\cite{Ladner77}.
\item \emph{Modal logic $T$:} The logic $T$ has the same syntax as
  $K$. Its coalgebraic semantics is a structure over the
  copointed functor $R$ with signature functor $\powerset$, given by 
  \begin{equation*}
    RX=\{(A,x)\in\powerset X\times X\mid x\in A\}.
  \end{equation*}
  Thus, $R$-coalgebras are reflexive Kripke frames. The interpretation
  of $\Box$ is defined as for $K$. (Axiomatically, $T$
  is determined by the non-iterative axiom $\Box a\modimpl a$.)  
\item\label{item:ck} \emph{Conditional logic $\CK$:} The signature of
  conditional logic has a single binary modal operator $\Rightarrow$,
  written in infix notation. Formulas $\phi\Rightarrow\psi$ are read
  as non-monotonic conditionals. The semantics of the conditional
  logic $\CK$~\cite{Chellas80} is given by a simple structure over the
  functor $\CF$ given by $\CF(X)=(\contrapow X\to \powerset X)$, with
  $\to$ denoting function space and $\contrapow$ contravariant
  powerset, cf.\ Definition~\ref{def:lifting}. $\CF$-coalgebras are
  \emph{conditional frames}~\cite{Chellas80}. The operator
  $\Rightarrow$ is interpreted over~$\CF$ by
  \begin{equation*}
    \Sem{\Rightarrow}_X(A,B)=
    \{f:\contrapow X\to\powerset X\mid f(A)\subseteq B\}.
  \end{equation*}
\item \emph{Conditional logic \CKId:} The conditional logic
  \CKId~\cite{Chellas80} extends $\CK$ with the rank-1 axiom
  $a\Rightarrow a$, referred to as $\ID$. The semantics of \CKId is
  modelled by restricting the structure for $\CK$ to the subfunctor
  $\CFid$ of $\CF$ defined by
   \begin{equation*}
     \CFid(X)=\{f\in \CF(X)\mid \forall A\in\contrapow X.\,f(A)\subseteq A\}.
   \end{equation*}
\item \emph{Conditional Logic \CKmp:} The logic \CKmp~\cite{Chellas80}
  extends  $\CK$ with the non-iterative axiom
  \begin{equation*}
    (\MP)\quad (a\Rightarrow b)\modimpl (a \modimpl b).
  \end{equation*}
  (This axiom is undesirable in default logics, but reasonable in
  relevance logics.) Semantically, this amounts to passing from the
  functor $\CF$ to the copointed functor $\CFmp$ with signature functor
  $\CF$, defined by
  \begin{equation*}
    \CFmp(X) = \{(f,x)\in \CF(X)\times X\mid
    \forall A\subseteq X.\,x\in A\impl x\in f(A)\}.
  \end{equation*}
  %the modality $\Rightarrow$ is interpreted as for $\CK$.
\item \label{item:prob}\emph{Modal logics of quantitative
    uncertainty:} The modal signature of \emph{likelihood} has $n$-ary
  modal operators $\sum_{i=1}^na_il(\argument)\ge b$ for
  $a_1,\dots,a_1,b\in \Rat$.  The terms $l(\phi)$ are called
  \emph{likelihoods}. The interpretation of likelihoods varies. E.g.\
  the semantics of the \emph{modal logic of
    probability}~\cite{FaginHalpern94} is modelled coalgebraically by
  a structure over the (finite) distribution functor $\PDist$, where
  $\PDist X$ is the set of finitely supported probability
  distributions on $X$, and $\PDist f$ acts as image measure
  formation. Coalgebras for $\PDist$ are probabilistic transition
  systems (i.e.\ Markov chains).  Likelihoods are interpreted as
  probabilities; i.e.
  \begin{equation*}\textstyle
    \Sem{\sum_{i=1}^na_i\cdot l(\argument)\ge b}_X(A_1,\dots,A_n)=
    \{P\in\PDist X\mid
    \textstyle\sum_{i=1}^n a_iP(A_i)\ge b\}.
  \end{equation*}
  Alternatively, likelihoods may be interpreted as \emph{upper
    probabilities}~\cite{HalpernPucella02}, i.e.\ the functor $\PDist$
  is replaced by $\powerset\circ\PDist$, and in the above definition,
  $P(A_i)$ is replaced by $\probset^*(A_i)$, where for
  $\probset\in\powerset\PDist X$, $\probset^*(A)= \sup\ \{PA\mid
  P\in \probset\}$.  This setting describes situations where agents
  are unsure about the actual probability distribution.  Further
  alternative notions of likelihood include Dempster-Shafer belief
  functions and Dubois-Prade possibility
  measures~\cite{HalpernPucella02b}. An extension of the modal
  signature of likelihood is the modal signature of
  \emph{expectation}~\cite{HalpernPucella02b}, where instead of
  likelihoods one more generally considers expectations
  $e(\sum_{j=1}^{n_k}b_{ij}\phi_{ij})$. Here, linear combinations of
  formulas represent \emph{gambles}, i.e.\ real-valued outcome
  functions, where the payoff of $\phi$ is the characteristic function
  of $\phi$. The exact definition of expectation depends on the
  underlying notion of likelihood.

  One-step logics of quantitative uncertainty are often considered to be
  of independent interest. E.g.\ the one-step logic of probability,
  i.e.\ a logic without nesting of likelihoods that talks only about a
  single probability distribution, is introduced
  independently~\cite{FaginHalpernMegiddo90} and only later extended to
  a full modal logic~\cite{FaginHalpern94}. In fact, logics of
  expectation~\cite{HalpernPucella02b} and the logic of upper
  probability~\cite{HalpernPucella02} so far appear in the literature
  only as one-step logics; the corresponding modal logics are of
  interest as natural variations of the modal logic of probability.
\end{sparenumerate}
\end{expl}\noindent

\begin{conv}\label{conv:rep}
  We assume that $\Lambda$ is equipped with a size measure, thus
  inducing a size measure for $\FLang(\Lambda)$. For one-step formulas
  $\phi$ over $V$, we assume w.l.o.g.\ that $|V|\le\size(\phi)$.
  %Moreover, we need to be precise about representation sizes of
  %elements of $SX$. 
  For finite $X$, we assume given a representation of elements
  $(t,x)\in SX$ as strings of size $\size(t,x)$ over some finite
  alphabet. We do not require that all elements of $SX$ are
  representable, nor that all strings denote elements of $SX$. When
  $S$ is trivially copointed, we represent only $t\in S_0(X)$. We
  require that inclusions $SX\subseteq SY$ induced according to
  Assumption~\ref{ass:injective} by inclusions $X\subseteq Y$ into a
  finite base set $Y$
  %whose
  %elements are represented in size $\lceil\log|Y|\rceil$ (where
  %$\lceil r \rceil = \min \lbrace n \in \mathbb{N} \mid n \geq r
  %\rbrace$), 
  preserve representable elements and increase their size by at most
  $\log|Y|$.
%\begin{math}
%\size_Y(t,x)\le\lceil\log|Y|\rceil\size_X(t,x)
%\textrm{ for $(t,x)\in SX\subseteq SY$}.
%\end{math}
%This assumption is usually unproblematic: we may
%typically think of elements of $S_0Y$ as terms built over elements of $Y$,
%so that reinterpreting $t\in S_0X$ as $t\in S_0Y$ amounts to replacing
%elements of $X$ by the corresponding elements of~$Y$, possibly
%increasing their size to $\lceil\log|Y|\rceil$.
\end{conv}
\noindent We make these issues explicit for the above examples:
\begin{expl}\label{expl:rep}
\begin{sparenumerate}
\item \emph{Modal logics $K$ and $T$:} For $X$ finite, elements of
  $\powerset X$ are represented as lists of elements of $X$. 
  %For $X\subseteq Y,$ reinterpreting an
  %element of $\powerset X$ as an element of $\powerset Y$ amounts to
  %replacing a list over $X$ by the list of the corresponding elements of
  %$Y,$ so that Condition ($*$) of Convention~\ref{conv:rep} holds;
  %analogous arguments work for the remaining examples.
\item \label{item:rep-ck}\emph{Conditional logics $\CK$ and \CKId:} For
  $X$ finite, elements of $\contrapow X \to\powerset X$ are represented
  as partial maps $f_0:\contrapow X\rightharpoonup\powerset X$;
  such an $f_0$ represents the total map $f$ that extends $f_0$ by
  $f(A)=\emptyset$ when $f_0(A)$ is undefined. (The use of partial maps
  avoids exponential blowup.)
\item \emph{Conditional logic \CKmp:} For $X$ finite, a pair $(f_0,x)$
  consisting of a partial map
  $f_0:\contrapow X\rightharpoonup\powerset X$ and $x\in X$
  represents the pair $(f,x)$, where $f:\contrapow X\to\powerset X$
  extends $f_0$ by $f(A)=A\cap\{x\}$ in case $f_0(A)$ is undefined.
\item \emph{Modal logics of quantitative uncertainty:} Suitable compact
  representations are described
  in~\cite{FaginHalpernMegiddo90,HalpernPucella02,HalpernPucella02b}.
\end{sparenumerate}
\end{expl}\noindent

\section{Polynomially branching shallow models} \label{sec:ospmp}

\noindent We now turn to the announced construction of polynomially
branching shallow models for modal logics whose one-step logic has a
small model property; this construction leads to a \PSPACE
decision procedure.
\begin{defn}\label{def:ospmp}
  We say that $\Struct$ has the \emph{one-step polysize model property
    (OSPMP)} if there exist polynomials $p$ and $q$ such that,
  whenever a one-step pair $(\phi,\psi)$ over $V$ has a one-step model
  $(X,\tau,t,x)$, then it has a one-step model $(Y,\kappa,s,y)$ such
  that $|Y|\le p(|\psi|)$, $(s,y)$ is representable with
  $\size(s,y)\le q(|\psi|)$, and $y\in\kappa(a)$ iff $x\in\tau(a)$ for
  all $a\in V$.
\end{defn}\noindent
In analogy to the transition between rules and axioms described
in~\cite{Schroder06}, one-step pairs are interchangeable with one-step
formulas. In particular, we have
\begin{prop}\label{prop:ospmp-axiom} 
  The $\Lambda$-structure $\Struct$ has the OSPMP iff there exist
  polynomials $p$, $q$ such that, whenever a one-step formula $\psi$
  over $V$ has a one-step model $(X,\tau,t,x)$, then it has a one-step
  model $(Y,\kappa,s,y)$ such that $y\in\kappa(a)$ iff $x\in\tau(a)$
  for all $a\in V$, $|Y|\le p(|\psi_1|)$, and $(s,y)$ is representable
  with $\size(s,y)\le q(|\psi_1|)$, where $\psi\equiv\psi_1\sigma$
  with $\psi_1\in\Prop(\Lambda(W))$ and $\sigma$ a
  $\Prop(V)$-substitution.
\end{prop}
\begin{proof}
  \emph{Only if:} Let $(X,\tau,t,x)$ be a one-step model of a one-step
formula $\psi$ over $V$.  By Lemma~\ref{lem:pairs}, $\psi$ is equivalent
to a one-step pair of the form $(\phi,\psi_1)$, with $\psi_1$ as in the
statement. By the OSPMP, $(\phi,\psi_1)$ has a one-step model
$(Y,\kappa,s,y)$ such that $|Y|\le p(|\psi_1|)$, $\size(s,y)\le
q(|\psi_1|)$, and $y\in\kappa(a)$ iff $x\in\tau(a)$ for all $a\in V$; by
Lemma~\ref{lem:pairs}, this model gives rise to a one-step model of
$\psi$ with the components $Y$, $s$, $y$ unchanged.

\emph{If:} Let $(X,\tau,t,x)$ be a one-step model of a one-step pair
$(\phi,\psi)$ over $V$. By Lemma~\ref{lem:pairs}, $(\phi,\psi)$ is
equivalent to a one-step formula of the form $\psi\sigma$, where
$\sigma$ is a $\Prop(V)$-substitution. By assumption, $\psi\sigma$ has a
one-step model $(Y,\kappa,s,y)$ such that $|Y|\le p(|\psi|)$,
$\size(s,y)\le q(|\psi|)$, and $y\in\kappa(a)$ iff $x\in\tau(a)$ for all
$a\in V$. By Lemma~\ref{lem:pairs}, this model gives rise to a one-step
model of $(\phi,\psi)$ with the components $Y$, $s$, $y$ unchanged.
\end{proof}
\noindent
Both formulations of the OSPMP easily reduce to the case that
$\psi$ is a conjunctive clause.

\begin{rem}\label{rem:fmp}
It is shown in~\cite{Schroder06} that the one-step logic always has an
exponential-size model property: a one-step formula $\psi$ over $V$ has
a one-step model iff it has a one-step model with carrier set
$\powerset(V)$.
\end{rem}\noindent
We are now ready to prove the shallow model theorem.
\begin{defn}
A \emph{supporting Kripke frame} of an $S$-coalgebra $(X,\xi)$ is a
Kripke frame $(X,R)$ such that for each $x\in X$,
\begin{equation*}
\xi(x)\in S_0\{y\mid xRy\}\subseteq S_0X
\end{equation*}
(equivalently $(\xi(x),x)\in S\{y\mid xRy\}$).
%equivalently (xi(x),x) in S{y | xRy}: generally for Y subset X, t in
%S_0 Y, (t,x) in SY iff (t,x) in SX: only if by functoriality of S,
%if: retract r:X->Y, inclusion i:Y->X, then t = SrSit, thus (Sit,ix) in
%SX implies Sr(Sit,ix) = (t,x) in SY.
A state $x\in X$ is a \emph{loop} if $xRx$.
\end{defn}\noindent

\begin{thm}[Shallow model property]\label{thm:shallowmodels}
If $\Struct$ has the OSPMP, then $\FLang(\pls)$ has the
\emph{polynomially branching shallow model property}: There exist
polynomials $p$, $q$ such that every satisfiable $\FLang(\pls)$-formula
$\psi$ is satisfiable at the root of an $S$-coalgebra $(X,\xi)$
which has a supporting Kripke frame $(X,R)$ such that removing all loops
from $(X,R)$ yields a tree of depth at most $\rank(\psi)$ and branching
degree at most $p(|\psi|)$, and $(\xi(x),x)\in S\{y\mid xRy\}$ is
representable with $\size(\xi(x),x)\le q(|\psi|)$.
%Moreover, at $x_0$, the bounds on the branching degree can
%be sharpened to $p(|\psi_1|)$ and $q(|\psi_1|)$, respectively, where
%$\psi=\psi_1\sigma$ with $\psi_1\in\Prop(\Up(V))$ and $\sigma$ an
%$\FLang(\Lambda)$-substitution.
\end{thm}\noindent
\begin{defn}
For $x\in X$ and a $\powerset(X)$-valuation $\tau$, we put 
\begin{math}
\textstyle\Th_\tau(x)\equiv\Land_{x\in\tau(a)}a\conj
    \Land_{x\notin\tau(a)}\neg a.
\end{math}
\end{defn}\noindent
\begin{proof}[Proof of Theorem~\ref{thm:shallowmodels}]
Induction over the rank of $\psi$.  If $\rank(\psi)=0$, then $\psi$
evaluates to $\top$ and hence is satisfied in a singleton $S$-coalgebra
$(X,\xi)$, which exists by Assumption~\ref{ass:injective}.

Now let $\rank(\psi)=n+1$. Let $z_0$ be a state in an $S$-coalgebra
$(Z,\zeta)$ such that $z_0\modelsCA_{(Z,\zeta)}\psi$.  Let $\Sub(\psi)$
denote the set of subformulas of $\psi$ occuring in $\psi$ within the
scope of a modal operator, let $V$ be the set of variables $a_\rho$,
indexed over $\rho\in\Sub(\psi)$, and let $\sigma$ denote the
substitution taking $a_\rho$ to $\rho$ for all $\rho$. Let $\bar\psi$ be
the conjunction of all literals
$\epsilon\ModOp(a_{\rho_1},\dots,a_{\rho_n})$ such that
$\ModOp(\rho_1,\dots,\rho_n)$ is a subformula of $\psi$ and
$z_0\modelsCA_{(Z,\zeta)}\epsilon\ModOp(\rho_1,\dots,\rho_n)$. (Recall
that $\epsilon$ denotes either nothing or negation.) Moreover, let
$\phi$ denote the propositional theory of $\sigma$, i.e.\ the
conjunction of all clauses $\chi$ over $V$ such that $\chi\sigma$ is
$\Lang$-valid.

Then $(Z,\kappa,\zeta(z_0),z_0)$ is a one-step model of
$(\phi,\bar\psi)$, where $\kappa(a)=\Sem{\sigma(a)}_{(Z,\zeta)}$. By the
OSPMP, it follows that $(\phi,\bar\psi)$ has a one-step model
$(Y,\tau,t,x_0)$ of polynomial size in $|\bar\psi|$ such that for all
$\rho\in\Sub(\psi_1)$, $x_0\in\tau(a_\rho)$ iff $z_0\in\kappa(\rho)$,
which in turn is equivalent to $z_0\modelsCA_{(Z,\zeta)}\rho$.

From this model, we now construct a shallow model
$(X,\xi)$ for $\psi$.  To begin, note that $\Th_\tau(y)\sigma$ is
$\Lang$-satisfiable for every $y\in Y$. For suppose not; then
$\neg\Th_\tau(y)\sigma$ is $\Lang$-valid, hence $\neg\Th_\tau(y)$ is a
conjunct of $\phi$. Thus, $Y,\tau\modelsPL\neg\Th_\tau(y)$, in
contradiction to the fact that $y\in\Sem{\Th_\tau(y)}\tau$ by
construction.  By induction, we thus have, for every $y\in Y$, a shallow
model $(X_y,\xi_y)$ of $\Th_\tau(y)\sigma$, where we may assume
$y\in X_y$ and $y\modelsCA_{(X_y,\xi_y)}\Th_\tau(y)\sigma$, with
depth at most $\rank(\Th_\tau(y)\sigma)=n$. We take $(X,\xi)$ as
the disjoint union of the $(X_y,\xi_y)$ over $y\in Y-\{x_0\}$,
extended by the state $x_0$, for which we put $\xi(x_0)=t\in S_0Y\subseteq
S_0X$.

We have to verify that $x_0\modelsCA_{(X,\xi)}\psi$. We will prove the
stronger statement $x_0\modelsCA_{(X,\xi)}\bar\psi\sigma$, i.e.\
\begin{equation}\label{eq:goal}
t\modelsOS_{X,\theta} \bar\psi,
\end{equation}
where $\theta(a_\rho)=\Sem{\rho}_{(X,\xi)}$ for
$\rho\in\Sub(\psi)$.

By induction over $\chi$ and naturality of predicate liftings,
$y\modelsCA_{(X,\xi)}\chi$ iff
$y\modelsCA_{(X_y,\xi_y)}\chi$ for $y\in Y-\{x_0\}$ and for
every formula $\chi$. In particular,
$y\modelsCA_{(X,\xi)}\Th_\tau(y)\sigma$ for all $y\in Y-\{x_0\}$,
i.e.\
\begin{equation}\label{eq:y-x}
y\modelsCA_{(X,\xi)}\rho\iff y\in\tau(a_\rho)
\end{equation}
for all $\rho\in\Sub(\psi)$. We prove by induction over
$\rho\in\Sub(\psi)$ that
\begin{equation}\label{eq:x}
x_0\modelsCA_{(X,\xi)}\rho\iff x_0\in\tau(a_\rho),
\end{equation}
which in connection with~(\ref{eq:y-x}) yields
\begin{equation}\label{eq:tau}
\Sem{\rho}_{(X,\xi)}\cap Y = \tau(a_\rho).
\end{equation}
The steps for boolean operations are straightforward. For
$\ModOp(\rho_1,\dots,\rho_n)\in\Sub(\psi)$, we have
\begin{align*}
&  x_0\modelsCA_{(X,\xi)}\ModOp(\rho_1,\dots,\rho_n)\\
\iff & t\in\Sem{\ModOp}_Y(\Sem{\rho_i}_{(X,\xi)}\cap
Y)_{i=1,\dots,n}= \Sem{\ModOp}_Y(\tau(a_{\rho_1}),\dots,\tau(a_{\rho_n}))\\
\iff & t\modelsOS_{(Y,\tau)} \ModOp(a_{\rho_1},\dots,a_{\rho_n}),
\end{align*}
using naturality of $\Sem{\ModOp}$ in the first step and the inductive
hypothesis in the shape of (\ref{eq:tau}) in the subsequent
equality. Since $t\models_{(Y,\tau)}\bar\psi$, the last statement is
equivalent to $z_0\modelsCA_{(Z,\zeta)} L(\rho_1,\dots,\rho_n)$. By the
definition of $\kappa$, this is equivalent to
$z_0\in\kappa(a_{L(\rho_1,\dots,\rho_n)})$, which in turn is equivalent
to $x_0\in\tau(a_{L(\rho_1,\dots,\rho_n)})$ by construction of
$(Y,\tau,t,x_0)$.

By~(\ref{eq:tau}) and naturality of predicate liftings, our remaining
goal~(\ref{eq:goal}) reduces to $t\modelsOS_{Y,\tau}\bar\psi$, which holds
by construction.

Finally, we have to establish that the overall branching degree of the
model is polynomial in $|\psi|$. The model is recursively constructed
from polynomial-size one-step models for pairs whose second components
are conjunctive clauses over atoms $L(a_{\rho_1},\dots,a_{\rho_n})$,
where $L(\rho_1,\dots,\rho_n)$ is a subformula of $\psi$. Such
conjunctive clauses are of at most quadratic size in $|\psi|$ (even
$O(|\psi|\log|\psi|)$ if subformulas of $\psi$ are represented by
pointers into $\psi$); this proves the claim.
\end{proof}

\begin{rem}\label{rem:ospmp}
%The leaves $x$ of the shallow model constructed above, i.e.\ the
%satisfying states for formulas of rank~$0$, are loops in the supporting
%Kripke frame in case the element $\xi(x)\in T\{x\}$ chosen in the base
%step fails to be in $T\emptyset\subseteq T\{x\}$. In particular, this is
%necessarily the case if $T\emptyset=\emptyset$ as e.g.\ in the case
%$T=\PDist$. 
While it is to be expected that the construction of polynomially
branching models depends on a condition like the OSPMP, it does not seem
to be the case that the precise formulation of this condition is
implicit in the literature (not even for the trivially copointed case).
Note in particular that the polynomial bound depends only on the second
component of a one-step pair. This is crucial, as the first component of
the one-step pair constructed in the above proof may be of exponential
size.  When we say in the introduction that the OSPMP can be obtained
from off-the-shelf results
(e.g.~\cite{FaginHalpernMegiddo90,HalpernPucella02,HalpernPucella02b}),
we refer to polynomial-size model theorems in which the polynomial bound
depends, in the notation of Proposition~\ref{prop:ospmp-axiom}, on
$|\psi|$, which may be exponentially larger than $|\psi_1|$; typically,
only an inspection of the given proofs shows that the bound can be
sharpened to be polynomial in~$|\psi_1|$.
\end{rem}\noindent
The proof of Theorem~\ref{thm:shallowmodels} leads to the following
nondeterministic decision procedure.
\begin{algorithm}\label{alg:pspace}
(Decide satisfiability of an $\FLang(\pls)$-formula $\psi$) Let
$\Struct$ have the OSPMP, and let $p$, $q$ be polynomial bounds as in
Definition~\ref{def:ospmp}. 
\begin{algenumerate}
\item If $\rank(\psi)=0$, terminate successfully if $\psi$ evaluates to
$\top$, else unsuccessfully. Otherwise:
\item \label{step:osformula} Take $V$ and $\sigma$ as in the proof of
  Theorem~\ref{thm:shallowmodels}, and guess a conjunctive clause
  $\bar\psi$ over $\Lambda(V)$ containing for each subformula
  $L(\rho_1,\dots,\rho_n)$ of $\psi$ either
  $L(a_{\rho_1},\dots,a_{\rho_n})$ or $\neg
  L(a_{\rho_1},\dots,a_{\rho_n})$ such that $\bar\psi\sigma$
  propositionally entails $\psi$.
\item \label{step:polybranch} Guess a $\powerset(Y)$-valuation $\tau$
  for $V$ and $(t,x)\in SY$ with $\size(t,x)\le q(|\bar\psi|)$, where
  $Y=\{1,\dots,p(|\bar\psi|)\}$, such that
  $t\modelsOS_{Y,\tau}\bar\psi$.
\item \label{step:recurse} For each $y\in Y$, check recursively that
  $\Th_\tau(y)\sigma$ is satisfiable.
\end{algenumerate}
\end{algorithm}
\noindent Since the rank decreases with each recursive call, the above
algorithm can be implemented in polynomial space, provided that
Step~\ref{step:polybranch} can be performed in polynomial space.
\begin{defn}
  The \emph{one-step model checking problem} is to check, given a
  string $s$, a finite set $X$, $A_1,\dots,A_n\subseteq X$, and
  $L\in\Lambda$ $n$-ary, whether $s$ represents some $(t,x)\in SX$ and
  whether $t\in\Sem{L}_X(A_1,\dots,A_n)$.
%The input size of
%this problem is $1+\size(s)+\size(L)+n|X|$.
\end{defn}\noindent
This property and the above algorithm lead to a \PSPACE bound for the
modal logic. Moreover, for bounded-rank fragments, the polynomially
branching shallow model property becomes a polynomial size model
property, thus leading to an \NP upper bound:
\begin{cor}\label{cor:complexity}
Let $\Struct$ have the OSPMP.
\vspace{-0.7em}
\begin{sparenumerate}
\item \label{item:pspace} If one-step model checking is in \PSPACE,
  then the satisfiability problem of $\FLang(\pls)$ is in \PSPACE.
\item \label{item:np} If one-step model checking is in $P$, then the
  satisfiability problem of $\FLang_n(\pls)$ is in \NP for every
  $n\in\Nat$.
\end{sparenumerate}
\end{cor}\noindent
%\begin{rem}
%A more refined analysis shows that for the \PSPACE upper bound, it
%suffices to require the one-step model checking problem to be in the
%polynomial hierarchy; however, we presently do not see practical
%applications of this extra generality.
%\end{rem}\noindent
\begin{proof}
\emph{1:} By Algorithm~\ref{alg:pspace}.

\emph{2:} Let $p$ and $q$ be polynomial bounds on the branching degree of
supporting Kripke frames and on the size of successor structures
$\xi(x)$ as guaranteed by Theorem~\ref{thm:shallowmodels}. Let
$n\in\Nat$.  Then by Theorem~\ref{thm:shallowmodels}, every satisfiable
formula $\psi\in\FLang_n(\Lambda)$ is satisfiable in a model
$(X,\xi)$ such that $|X|\le\sum_{i=0}^np(|\psi|)^i=:N$ and
$\size(\xi(x))\le \log(N)q(|\psi|)$ for all $x\in X$, where the second
inequality relies also on Convention~\ref{conv:rep}. Thus, the entire
representation size of the model $(X,\xi)$ is bounded by
$M:=N\log(N) q(|\psi|)$, which is polynomial in $|\psi|$. Thus, the
following non-deterministic algorithm decides satisfiability of $\psi$
in polynomial time:
\begin{sparenumerate}
\item Guess a model $(X,\xi)$ of size at most $M$
\item Check that $(X,\xi)$ is an $S$-coalgebra.
\item Check that $\Sem{\psi}_{(X,\xi)}\neq\emptyset$.
\end{sparenumerate}
\noindent 
The second step can be performed in polynomial time because one-step
model checking is in $P$. The third step can be performed in polynomial
time by recursively computing extensions $\Sem{\phi}_{(X,\xi)}$, again
because one-step model checking is in $P$.
\end{proof}
\noindent This generalises results for the modal logics $K$ and $T$
established in~\cite{Halpern95}.
%(where corresponding results are proved also for several other normal
%modal logics outside rank~1)
\begin{expl}\label{expl:appl}
%For some of the examples below, namely the logics of upper probability
%and expectation, only the associated one-step logics are considered in
%the literature; these are in \NP by
%Proposition~\ref{prop:one-step-np} (and the cited work). However, the
%extension to state-dependency of belief effected by turning these logics
%into proper modal logics, as carried out in~\cite{FaginHalpern94} for
%the logic of probability introduced in~\cite{FaginHalpernMegiddo90}, is
%a natural next step in the modelling of belief. For clarity of
%terminology, we refer to the one-step logics as the \emph{logic} of
%upper probability etc.\ and to the state-dependent versions as the
%\emph{modal logic} of upper probability etc.
\begin{sparenumerate}
\item\label{item:appl-k} \emph{Modal logics $K$ and $T$:} One-step
  model checking for $K$ and $T$ amounts to checking a subset
  inclusion and, in the case of $T$, additionally an elementhood; this
  is clearly in $P$. To verify the OSPMP for $K$, let $(X,\tau,A)$ be
  a one-step model of a one-step pair $(\phi,\psi)$ over $V$;
  w.l.o.g.\ $\psi$ is a conjunctive clause over atoms $\Box a$, where
  $a\in V$. For $\neg\Box a$ in $\psi$, there exists $x_a\in A$ such
  that $x_a\notin\tau(a)$. Taking $Y$ to be the set of these $x_a$, we
  obtain a polynomial-size one-step model $(Y,\tau_Y,Y)$ of
  $(\phi,\psi)$, where $\tau_Y(a)=\tau(a)\cap Y$ for all $a$. The
  construction for $T$ is the same, except that the point $x$ of the
  original one-step model $(X,\tau,A,x)$ is retained in the carrier
  set $Y$, and becomes the point of the small model. By
  Corollary~\ref{cor:complexity}, this reproves Ladner's \PSPACE upper
  bounds for $K$ and $T$~\cite{Ladner77}, as well as Halpern's \NP
  upper bounds for bounded-rank fragments~\cite{Halpern95}.
\item \label{item:appl-ck} \emph{Conditional logic:} It is easy to see
  that one-step model checking for $\CK$, \CKId, and \CKmp is in~$P$. (In
  particular, deciding whether a given string represents an element of
  $\CFid(X)$ just amounts to checking subset inclusions. Moreover,
  deciding whether $(f,x)\in\CFmp(X)$, i.e.\ whether $x\in A$ implies
  $x\in f(A)$, can be done in polynomial time thanks to the choice of
  default value for $f$; cf.\ Example~\ref{expl:rep}.\ref{item:rep-ck}.)

  To prove that $\mi{\CK}$ has the OSPMP, let $(X,\tau,f)$ be a
  one-step model of a one-step pair $(\phi,\psi)$, where w.l.o.g.\
  $\psi$ is a conjunctive clause
  $\Land_{i=1}^n\epsilon_i(a_i\Rightarrow b_i)$. If
  $\tau(a_i)\neq\tau(a_j)$, fix an element $y_{ij}$ in the symmetric
  difference of $\tau(a_i)$ and $\tau(a_j)$. Moreover, if $\epsilon_i$
  is negation, fix $z_i\in f(\tau(a_i))\setminus\tau(b_i)$. Let $Y$ be
  the set of all $y_{ij}$ and all $z_i$. Let $\tau_Y$ be the
  $\powerset(Y)$-valuation defined by $\tau_Y(v)=\tau(v)\cap Y$, and
  let $f_Y\in \CF(Y)$ be represented by the partial map taking
  $\tau_Y(a_i)$ to $f(\tau(a_i))\cap Y$ for all $i$ (this is
  well-defined by construction of $Y$). Then $(Y,\tau_Y,f_Y)$ is a
  one-step model of $(\phi,\psi)$. The cardinality of $Y$ is quadradic
  in $\psi$, and the representation size of $f_Y$ is polynomial.

  % Viz. List of n pairs of sets, each a list of up to n^2 elements of
  % size log n each, thus n^3 log n.
  
  % CK+Id: need to verify f(tau(a_i))cap Y subseteq tau_Y(a_i) =
  % tau(a_i) cap Y, immediate from f(tau(a_i)) subseteq tau(a_i)

  Thanks to the choice of default value, this construction of
  polynomial-size one-step models works also for \CKId.  The
  construction for \CKmp is almost identical, except that the point $x$
  of $(X,\tau,f,x)$ is retained in the small one-step model
  $(Y,\tau_Y,f_Y,x)$; here, $(f_Y,x)\in\CFmp(Y)$ due to the different
  choice of default value.

%;to see that for $f\in \CFid(X)$, the small one-step 
  %model $(Y,\tau_Y,f_Y)$ indeed satisfies $f_Y\in \CFid Y$ note that 
  %$f(A)\subseteq A$ implies $f(A)\cap Y\subseteq A\cap Y$, and that 
  %$\emptyset\subseteq A$ for all $A$ (recall that $f_Y(A)=\emptyset$ if 
  %no image of $A$ is explicitly given in the representation of $f_Y$).  

  We thus obtain that $\mi{CK}$, \CKId, and \CKmp are in \PSPACE
  (hence \PSPACE-complete, as these logics contain~$K$ and --- in
  the case of \CKmp --- $T$, respectively, as sublogics). This has
  previously been proved using a detailed analysis of a labelled sequent
  calculus~\cite{OlivettiEA07} (the method of~\cite{OlivettiEA07} yields
  an explicit polynomial bound on space usage which is not matched by
  the generic algorithm). The \NP upper bound for bounded-rank
  fragments of $\CK$, \CKId, and \CKmp arising from
  Corollary~\ref{cor:complexity}.\ref{item:np} is, to our knowledge,
  new.
\item \label{item:appl-prob}\emph{Modal logics of quantitative
  uncertainty:} Polynomial size model properties for one-step logics
  have been proved for the logic of
  probability~\cite{FaginHalpernMegiddo90}, the logic of upper
  probability~\cite{HalpernPucella02}, and various logics of
  expectation~\cite{HalpernPucella02b}. As indicated in
  Remark~\ref{rem:ospmp}, the polynomial bounds are stated in the cited
  work as depending on the size of the entire one-step formula $\psi$;
  however, inspection of the given proofs shows that the polynomial
  bound in fact depends only on the number of likelihoods or
  expectations in $\psi$, respectively, and on the representation size
  of the largest coefficient. By Proposition~\ref{prop:ospmp-axiom}, it
  follows that the respective logics have the OSPMP. Suitable complexity
  estimates for one-step model checking are also found in the cited
  work.

  By the above results, it follows that the respective modal logics of
  quantitative uncertainty are in \PSPACE (hence
  \PSPACE-complete, as one can embed~$KD$ by mapping $\Diamond$ to
  $l(\argument)>0$), and in \NP when the modal nesting depth is
  bounded. For the modal logic of probability, a proof of the
  \PSPACE upper bound is sketched in~\cite{FaginHalpern94}. The
  \PSPACE upper bounds for the remaining cases (e.g.\ the modal
  logic of upper probability and the various modal logics of
  expectation) seem to be new, if only for the reason that only the
  one-step versions of these logics appear in the literature. Similarly,
  all \NP upper bounds for bounded-rank fragments are, to our
  knowledge, new. Moreover, the upper bounds extend easily to modal
  logics of uncertainty with non-iterative axioms, e.g.\ an axiom
  $a\modimpl l(a)\ge p$ which states that the present state remains
  stationary with likelihood at least~$p$.
\end{sparenumerate}
\end{expl}\noindent

\section{Extended Example: Elgesem's modal logic of agency}\label{sec:elgesem}

\noindent There have been numerous approaches to capturing the notion
of agents bringing about certain states of affairs, one of the most
recent ones being Elgesem's modal logic of agency
(\cite{Elgesem97}~and references therein,
\cite{GovernatoriRotolo05}). Modal logics of agency play a role e.g.\
in planning and task assignment in multi-agent systems (cf.\
e.g.~\cite{CholvyEA05,JonesParent07}).

Elgesem defines a logic with two modalities $E$ and $C$ (in general
indexed over agents; all results below easily generalise to the
multi-agent case), read `the agent brings about' and `the agent is
capable of realising', respectively. The semantics is given by a class
of conditional frames $(X,f:X\to\contrapow X\to\powerset X)$
(Example~\ref{expl:logics}.\ref{item:ck}), called \emph{selection
function models} in this context. The clauses for the modal
operators are
\begin{align*}
  x\models E\phi&\quad\textrm{iff}\quad x\in
  f(w)(\Sem{\phi})\quad\textrm{and}\\
  x\models C\phi&\quad\textrm{iff}\quad f(w)(\Sem{\phi})\neq\emptyset.
\end{align*}
The relevant class of selection function models $(X,f)$ is defined by the
conditions
\begin{equation*}
  \begin{axarray}
    \textrm{(E1)} & $f(x)(X)  = \emptyset$\\
    \textrm{(E2)} & $f(x)(A)\cap f(x)(B)\subseteq f(x)(A\cap B)$\\
    \textrm{(E3)} & $f(x)(A) \subseteq A$.
  \end{axarray}
\end{equation*}
It is shown in~\cite{Elgesem97,GovernatoriRotolo05} that the logic of
agency is completely axiomatised by $\neg C\top$, $\neg C\bot$, $Ea\land
Eb\modimpl E(a\land b)$, $Ea\modimpl a$, and $Ea\modimpl Ca$. Notably,
the agent is incapable of realising what is logically necessary ($\neg
C\top$), i.e.\ the notion of realising a state of affairs entails actual
attributability (this axiom is weaker than previous formulations using
\emph{avoidability}; cf.\ the baby food example
in~\cite{Elgesem97}). Monotonicity is not imposed. The axiom $\neg
C\bot$ is due to~\cite{GovernatoriRotolo05}.

Most of the information in selection function models (motivated by
philosophical considerations in~\cite{Elgesem97}) is irrelevant for
the semantics of $E$ and $C$: one only needs to know whether $f(x)(A)$
is non-empty, and whether it contains $x$. Moreover, the selection
function semantics fails to be coalgebraic, as the naturality
condition fails for the (generalised) predicate lifting implicit in
the clause for~$E$. Both problems are easily remedied by moving to the
following coalgebraic semantics: put $3=\{\bot,*,\top\}$ (to represent
the cases $f(x)(A)=\emptyset$, $x\notin f(x)(A)\neq\emptyset$, and
$x\in f(x)(A)$, respectively), and take as signature functor the
\emph{$3$-valued neighborhood functor} $\ThreeNb$ given by
\begin{math}
  \ThreeNb(X)=\contrapow(X)\to 3
\end{math}
(with $\contrapow(X)$ denoting contravariant powerset).  We define the
copointed functor $\AgFunct$ as the subfunctor of $\ThreeNb\times\Id$
such that $(f,x)\in \AgFunct(X)$ iff for all $A,B\subseteq X$,
\begin{equation*}
  \begin{axarray}
    \textrm{(E1$'$)} & $f(X)=\bot$\\
    \textrm{(E2$'$)} & $f(A)\meet f(B)\le f(A\cap B)$\\
    \textrm{(E3a$'$)} & $f(\emptyset)=\bot$\\
    \textrm{(E3b$'$)} & $f(A)=\top \implies x\in A$    
  \end{axarray}
\end{equation*}
where $\meet$ and $\le$ refer to the ordering $\bot<*<\top$.  We define
a structure over $\AgFunct$ for the modal logic of agency by
\begin{align*}
  \Sem{E}_XA & =\{f:\contrapow\to 3\mid f(A)=\top\}\\
  \Sem{C}_XA & =\{f:\contrapow\to 3\mid f(A)\neq\bot\}.
\end{align*}
\begin{prop}\label{prop:agency-models}
  A formula of the modal logic of agency is satisfiable in a selection
  function model iff it is satisfiable over $\AgFunct$.
\end{prop}
\begin{proof}
    \emph{`Only if:'} Given a selection function model $(X,f)$, define an
  $\ThreeNb$-coalgebra $(X,\tilde f)$ by 
  \begin{equation*}
    \tilde f(x)(A)= \begin{cases}\top &\textrm{if $x\in f(x)(A)$}\\
      * & \textrm{if $x\notin f(x)(A)\neq\emptyset$}\\
      \bot & \textrm{if $f(x)(A)=\emptyset$}.
      \end{cases}
  \end{equation*}
  It is clear that $(X,\tilde f)$ is an $\AgFunct$-coalgebra and that
  $x\in X$ satisfies the same formulas in $(X,\tilde f)$ as in $(X,f)$.

  \emph{'If':} Let $(X,f)$ be an $\AgFunct$-coalgebra. We can assume
    that $|\Sem{\phi}_{(X,f)}|\neq 1$ for all formulas $\phi$
    (otherwise, form the coproduct of $(X,f)$ with itself, so that each
    state has a twin satisfying the same formulas). We define a
    selection function model $(X,\bar f)$ by
    \begin{equation*}
      \bar f(x)(A)=\begin{cases}
      A & \textrm{if $f(x)(A)=\top$}\\
      A-\{x\} & \textrm{if $f(x)(A)=*$}\\
      \emptyset & \textrm{if $f(x)(A)=\bot$}.
      \end{cases}
    \end{equation*}
    It is clear that $(X,\bar f)$ satisfies \textrm{E1}--\textrm{E3}.  
    %E1: bar f(x)(X)=empty, da f(x,X)=bot 
    % E2: klar, falls beides bot. Sonst: 
    % Fall a: o.E. f(x)(A)=*, f(x)(B)>=*, 
    % dann f(x)(A cap B)>=* per (E2').  
    % Somit bar f(x)(A) cap bar f(x)(B) subset A-x cap B = A cap B -x 
    % subset f(x)(A cap B) % Fall b: beides gleich top, dann auch 
    % f(x)(A cap B), also Gl. A cap B subset A cap B 
    % E3: by construction 
    One shows by induction over the formula structure that
    $x\in X$ satsifies the same formulas in $(X,\bar f)$ as in $(X,f)$,
    with the only non-trivial point being that in the step for the modal
    operator $C$, one has to note that, by the above assumption,
    $\Sem{\phi}_{(X,f)}-\{x\}\neq\emptyset$ whenever
    $f(x)(\Sem{\phi}_{(X,f)})=*$.  
\end{proof}
\noindent
To avoid exponential explosion, we represent elements of $\ThreeNb(X)$,
for $X$ finite, using partial maps $f_0:\contrapow(X)\parr 3$.
%represented as lists of pairs $(A,b)$,
%where $A$, in turn, is represented as a list of elements. 
To enforce~(E2$'$), we let such an $f_0$ represent the map
$f:\contrapow(X)\to 3$ that maps $B\subseteq X$ to the maximum of
\begin{math}
\Meet_{i=1}^n f_0(A_i),
\end{math}
taken over all sets $A_1,\dots,A_n\subseteq X$ such that $\bigcap A_i=B$
and $f_0(A_i)$ is defined for all~$i$; when no such sets exist, the
maximum is understood to be $\bot$.
\begin{lem}\label{lem:elgesem-represent}
  Let $f_0$ and $f$ be as above.
  \begin{sparenumerate}
    \item\label{item:lower} Whenever $f_0(A)$ is defined, then
      $f_0(A)\le f(A)$.
    \item\label{item:upper} Let $b\in 3$. Then $f(A)\ge b$ iff
      \begin{math}
	\bigcap\{B\subseteq X\mid A\subseteq B,f_0(B)\ge b\textrm{
	  defined}\} = A.
      \end{math}
    \item\label{item:E1} The pair $(f,x)$ satisfies (E1$'$)
      iff $f_0(X)$ is either undefined or equals $\bot$.
    \item \label{item:E2} The pair $(f,x)$ satisfies (E2$'$).
    \item \label{item:E3a} The pair $(f,x)$ satisfies (E3a$'$) iff 
      \begin{math}
	\bigcap\{A\subseteq X\mid f_0(A)>\bot\textrm{
	defined}\}\neq\emptyset.
      \end{math}
    \item\label{item:E3b} The pair $(f,x)$ satisfies (E3b$'$)
      iff whenever $f_0(A)=\top$ is defined, then $x\in A$.
  \end{sparenumerate}
\end{lem}
\begin{proof}
  
\noindent\emph{\ref{item:lower}.:} Trivial.

 \noindent\emph{\ref{item:upper}.:} `If' is trivial. `Only if': by assumption,
   $A=\bigcap_{i=1}^n B_i$ for some $B_i$ such that $f_0(B_i)\ge b$ is
   defined for all $i$; the claim follows immediately.
  
\noindent\emph{\ref{item:E1}.:}  `Only if' is immediate by~\ref{item:lower}.,
    and `if' holds because $X=\bigcap A_i$ implies $A_i=X$ for all $i$.

 \noindent\emph{\ref{item:E2}.:} By construction.

\noindent\emph{\ref{item:E3a}.:} Immediate by~\ref{item:upper}.
  
\noindent\emph{\ref{item:E3b}.:} `Only if' holds by 1., and `if' holds because
  $f(B)=\top$ implies that $B=\bigcap A_i$ for sets $A_i$ such that
  $f_0(A_i)=\top$ for all $i$.
\end{proof}
\noindent
By Lemma~\ref{lem:elgesem-represent}, it is immediate that one-step
model checking is in $P$.  To prove the OSPMP, let
$(X,\tau,f:\contrapow(X)\to 3,x)$ be a one-step model of a one-step pair
$(\phi,\psi)$ over $V$. By Remark~\ref{rem:fmp}, we can assume that $X$
is finite. Let the set $Y\subseteq X$ consist of
\begin{sparitemize}
  \item the element $x$;
  \item an element $y_{ab}\in\tau(a)\setminus\tau(b)$ for each pair
    $(a,b)\in V^2$ such that $\tau(a)\not\subseteq\tau(b)$;
  \item an element 
    \begin{math}
      z_a\in\bigcap\,\{\tau(b)\mid b\in V,\tau(a)\subseteq \tau(b),
      \end{math} \begin{math}
      f(\tau(b))>f(\tau(a))\}\setminus \tau(a)
    \end{math}
    for each $a\in V$ ($z_a$ exists by (E2$'$)); and
  \item an element $w_0\in\bigcap\{\tau(b)\mid f(\tau(b))>\bot\}$ ($w_0$
    exists by (E2$'$) and (E3a$'$)).
\end{sparitemize}
Put $\tau_Y(a)=\tau(a)\cap Y$ for $a\in V$, and let $f_Y$ be represented
by the partial map $f_0$ taking $\tau_Y(a)$ to $f(\tau(a))$ ($f_0$
is well-defined by construction of $Y$). Then $Y$ and $(f_Y,x)$ are of
polynomial size in $\psi$, and $(Y,\tau_Y)\models\phi$. By
Lemma~\ref{lem:elgesem-represent}, $(f_Y,x)$ is in $\AgFunct(X)$, with
the criterion for (E3a$'$) satisfied due to $w_0\in Y$. By
Lemma~\ref{lem:elgesem-represent}.\ref{item:upper}, the $z_a\in Y$
ensure that $f_Y(\tau_Y(a))=f(\tau(a))$ for all $a\in V$, so that
$f_Y\modelsOS_{(Y,\tau_Y)}\psi$.

By Corollary~\ref{cor:complexity}, we obtain that \emph{the modal
  logic of agency is in \PSPACE{}}, and that bounding the modal
nesting depth brings the complexity down to \NP. Both results (and
even decidability) seem to be new. In the light of the previous
observation that the agglomeration axiom $Ea\land Eb\modimpl E(a\land
b)$ tends to cause \PSPACE-hardness~\cite{Vardi89}, we conjecture that
the \PSPACE upper bound is tight.
% Halpern/Rego treat only the case that includes K.

\section{Exponential Branching}

\noindent In cases where the OSPMP fails, it may still be possible to
obtain a \PSPACE upper bound by traversing an exponentially
branching shallow model (by Remark~\ref{rem:fmp}, branching is never
worse than exponential). The crucial prerequisite is that
exponential-size one-step models can be traversed pointwise,
accumulating during the traversal a polynomial amount of information
that suffices for one-step model checking. This requires additional
assumptions on the signature functor $S_0$:
\begin{defn}
  We say that $S_0$ is \emph{pointwise bounded} w.r.t.\ a set $C$ if
  for all sets $X$, there exists an injection $S_0X\into (X\to C)$
  (i.e.\ $|S_0X|\le|X\to C|$). We then identify $S_0X$ with a subset
  of $X\to C$.
\end{defn}
\noindent (Note that the above does \emph{not} require that $\lambda
X.\,X\to C$ is functorial.)  Recall from~\cite{Schroder05} that the
signature functor $S_0$ admits a \emph{separating} set of unary
predicate liftings (separation is a necessary condition for the
generalised Hennessy-Milner property) iff the family of maps
$S_0f:S_0X \to S_02$, indexed over all maps $f:X\to 2=\{\bot,\top\}$,
is jointly injective for each set $X$. Typically, functors $S_0$
satisfying this condition satisfy the stronger requirement that
already the family of maps
\begin{math}
  (S_0\cf_{\{x\}}:S_0X\to S_02)_{x\in X}
\end{math}
is jointly injective, where $\cf_A$ denotes the characteristic
function of $A\subseteq X$, so that $S_0$ is pointwise bounded w.r.t.\
$S_02$; often, even a quotient of $S_02$ will suffice. Of the signature
functors mentioned in Example~\ref{expl:logics}, $\powerset$ and
$\PDist$ are pointwise bounded (w.r.t.\ $2$ and $[0,1]$,
respectively), while $\CF$ and $\powerset\circ\PDist$ fail to be
so. Further examples of pointwise bounded functors include the game
frame functor appearing in the semantics of coalition
logic~\cite{SchroderPattinson06} and the multiset functor introduced
below.

Assume from now on that $S_0$ is pointwise bounded w.r.t.\ $C$, with a
given representation of elements of $C$ (Convention~\ref{conv:rep} is
no longer needed). For $t:X\to C$, we define
\begin{equation*}
  \maxsize(t)=\max_{x\in X}\size(t(x)),
\end{equation*}
and put $\maxsize(t,x)=\maxsize(t)$ for $x\in X$.
\begin{defn}\label{def:wk-ospmp}
  We say that $\Struct$ has the \emph{one-step pointwise polysize
    model property (OSPPMP)} if there exists a polynomial $p$ such
  that, whenever a one-step pair $(\phi,\psi)$ over $V$ has a one-step
  model $(X,\tau,t,x)$, then it has a one-step model $(Y,\kappa,s,y)$
  such that $|Y|\le 2^{|V|}$, $\maxsize(s)\le p(|\psi|)$, and
  $y\in\kappa(a)$ iff $x\in\tau(a)$ for all $a\in V$; such a model is
  called \emph{pointwise polysize}.
\end{defn}\noindent
By Remark~\ref{rem:fmp}, the actual content of the OSPPMP is the
polynomial bound on $\maxsize(s)$. The OSPPMP holds for all
pointwise bounded functors mentioned so far, trivially so in cases
where $C$ is finite. We have a variant of
Theorem~\ref{thm:shallowmodels}, proved entirely analogously, which
states that \emph{under the OSPPMP, every satisfiable formula $\psi$ is
satisfied in a shallow model $(X,\xi)$ with branching degree at most
$2^{|\psi|}$ and $\maxsize(\xi(x))$ polynomially bounded in $|\psi|$}.
% \begin{thm}\label{thm:exp-shallowmodels}
%   If $\Struct$ has the OSPPMP, then there exists a polynomial~$p$
%   such that every satisfiable $\FLang(\pls)$-formula $\psi$ is
%   satisfiable at the root of an $S$-coalgebra $(X,\xi)$ which has a
%   supporting Kripke frame $(X,R)$ such that removing all loops from
%   $(X,R)$ yields a tree of depth at most $\rank(\psi)$, branching
%   degree at most $2^{|\psi|}$, and $\maxsize(\xi(x))\le p(|\psi|)$ for
%   all $x\in X$.
% %Moreover, at $x_0$, the bounds on the branching degree can
% %be sharpened to $p(|\psi_1|)$ and $q(|\psi_1|)$, respectively, where
% %$\psi=\psi_1\sigma$ with $\psi_1\in\Prop(\Up(V))$ and $\sigma$ an
% %$\FLang(\Lambda)$-substitution.
%\end{thm}\noindent
For the ensuing algorithmic treatment, we need a refined notion of
one-step model checking:
\begin{defn}\label{def:posmc}
  The \emph{pointwise one-step model checking problem} is to check,
  given a map $t:X\to C$, $x\in X$, a $\powerset(X)$-valuation $\tau$
  for $V$, $Y\subseteq X$, and a conjunctive clause $\psi$ over
  $\Lambda(V)$, whether $(t,x)\in SY\subseteq (X\to C)\times X$ and
  $t\modelsOS_{Y,\tau_Y}\psi$, where $\tau_Y(a)=\tau(a)\cap Y$ for
  $a\in V$. We say that this problem is \emph{\PSPACE-tractable} if it
  is decidable on a non-deterministic Turing machine with input tape
  that uses space polynomial in $\maxsize(t)$ and accesses each input
  symbol at most once.
\end{defn}
\begin{thm}\label{thm:alt-pspace}
  If $\Struct$ has the OSPPMP and pointwise one-step model checking
  is \PSPACE-tractable, then the satisfiability problem of
  $\FLang(\Lambda)$ is in \PSPACE.
\end{thm}

\begin{proof}
  Let $M$ be a decision procedure for pointwise one-step
  model checking as required in the definition of \PSPACE-tractability
  (Defn.~\ref{def:posmc}). Let $p$ be a polynomial witnessing the
  OSPPMP as in Definition~\ref{def:wk-ospmp}. Then the following
  non-deterministic algorithm decides satisfiability of
  $\FLang(\Lambda)$-formulas:
\begin{algorithm}\label{alg:alt-pspace}
\begin{algenumerate}
\item If $\rank(\psi)=0$, terminate successfully if $\psi$ evaluates to
$\top$, else unsuccessfully. Otherwise:
\item  Take $V$ and $\sigma$ as in the proof of
  Theorem~\ref{thm:shallowmodels},
  %(which coincides almost entirely
  %with the omitted proof of the relevant
  %Theorem~\ref{thm:exp-shallowmodels}), 
  and guess a conjunctive clause
  $\bar\psi$ over $\Lambda(V)$ containing for each subformula
  $L(\rho_1,\dots,\rho_n)$ of $\psi$ either
  $L(a_{\rho_1},\dots,a_{\rho_n})$ or $\neg
  L(a_{\rho_1},\dots,a_{\rho_n})$ such that $\bar\psi\sigma$
  propositionally entails $\psi$.
\item \label{step:posmc} Call $M$ with arguments $X,\tau,Y,t$ as in
  Definition~\ref{def:posmc} to check that $t\in TY$ and
  $t\modelsOS_{(Y,\tau_Y)}\bar\psi$, with $\tau_Y$ as in
  Definition~\ref{def:posmc}. Here, $X=2^V$, $\tau(a)=\{B\in X\mid
  a\in B\}$, $$Y=\{B\in X\mid \Land_{a_{\rho}\in
    B}\rho\land\Land_{a_{\rho}\notin B}\neg\rho\,\textrm{
    satisfiable}\}$$ is calculated recursively, and $t\in (X\to C)$
  with $\maxsize(t)\le p(|\phi|)$ is guessed.
\end{algenumerate}
\end{algorithm}
\noindent It remains to see that the above algorithm can be
implemented in polynomial space although the input to $M$ in
Step~\ref{step:posmc} is of overall exponential size. This is achieved
by replacing read operations on the input tape in $M$ by calls to a
procedure passed by the caller, which produces the \mbox{$k$-th} input
symbol on demand, and then calling the modified pointwise model
checker $M'$ with such a procedure instead of the full argument. By
the assumption that $M$ accesses each symbol on the input tape at most
once, there is no need to keep the symbols representing the guessed
value~$t$ in memory after they have been passed to $M'$. Therefore,
only polynomial space overhead is generated by the input to $M'$ (the
input procedure depends on $\phi$ and hence has representation size
$O(|\phi|)$); the space usage of $M'$ itself is polynomial in
$\maxsize(t)$ and therefore in $|\phi|$. 
\end{proof}
\begin{expl}
  Theorem~\ref{thm:alt-pspace} applies e.g.\ to the modal logics $K$
  and $T$, as well as to probabilistic modal logic; however, as all
  these logics in fact enjoy the OSPMP, the method of
  Sect.~\ref{sec:ospmp} is preferable in these cases. A more
  interesting application is given by graded modal
  logic~\cite{Fine72}, or more generally Presburger modal
  logic~\cite{DemriLugiez06}. 

  In its single-agent version, Presburger modal logic has $n$-ary
  modal operators $\sum_{i=1}^na_i\#(\argument)\sim b$, where $b$ and
  the $a_i$ are integers and $\sim\,\in\{<,>,=\}\cup\{\equiv_k\mid
  k\in\Nat\}$. A coalgebraic semantics for this logic, equivalent for
  purposes of satisfiability to the ordered tree semantics given
  in~\cite{DemriLugiez06}, is defined over the finite multiset functor
  $\Bag$, which maps a set~$X$ to the set of maps $B:X\to\Nat$ with
  finite support, the intuition being that~$B$ is a multiset
  containing $x\in X$ with multiplicity~$B(x)$.   For $A\subseteq X$,
  put $B(A)=\sum_{x\in A}B(x)$. $\Bag$-coalgebras are graphs with
  $\Nat$-weighted edges. The above modalities are interpreted by
  \begin{equation*}\textstyle
    \Sem{\sum_{i=1}^na_i\#(\argument)\sim b}_X(A_1,\dots,A_n)=
    \{B\in\Bag(X)\mid \textstyle\sum_{i=1}^na_iB(A_i)\sim b\},
  \end{equation*}
  with $>,<,=$ interpreted as expected, and $\equiv_k$ as equality
  modulo~$k$. This logic extends graded modal logic, whose operators
  $\gldiamond{k}$ now become $\#(\argument)>k$.

  Of course, $\Bag$ is pointwise bounded w.r.t.\ $\Nat$. It follows
  easily from estimates on solution sizes of integer linear
  equalities~\cite{Papadimitriou81} that Presburger modal logic has
  the OSPPMP~\cite{DemriLugiez06}. Moreover, given a conjunctive
  clause $\psi$ over $\Lambda(V)$, a $\powerset(X)$-valuation $\tau$,
  and $B\in\Bag(X)$, one can check whether $B\modelsOS_{X,\tau}\psi$
  by traversing $X$ and computing the $B(\tau(a))$ by successive
  summation; it is thus easy to see that pointwise one-step model
  checking is \PSPACE-tractable. It follows that \emph{Presburger
    modal logic is in \PSPACE{}}. While this is proved already
  in~\cite{DemriLugiez06}, using essentially the same type of
  algorithm\footnote{\noindent The claim that a (rank-1) logic further
    extended by regularity constraints is still in \PSPACE is
    retracted in the full version of~\cite{DemriLugiez06} as being
    based on possibly erroneous third-party results.}, our method
  extends straightforwardly to extensions of Presburger modal logic by
  certain frame conditions such as reflexivity (modelled by the
  copointed functor $SX=\{(B,x)\in\Bag X\times X\mid B(x)>0\}$) or
  e.g.\ the condition that at least half of all transitions from a
  given state are loops (modelled by the copointed functor
  $SX=\{(B,x)\in\Bag X\times X\mid B(x)\ge B(X-\{x\}\}$). In
  particular, this implies that \emph{graded modal logic over
    reflexive frames} (i.e.\ the logic $Tn$ of~\cite{Fine72}) \emph{is
    in \PSPACE{}}, to our knowledge a new result. The logic $Tn$ can
  be seen as a description logic with qualified number restrictions on
  a single reflexive role. Our arguments extend straightforwardly to
  show that a description logic with role hierarchies, reflexive
  roles, and qualified number restrictions has concept satisfiability
  over the empty $T$-box in \PSPACE. As reflexivity of a role
  $R$ is expressed by the role inclusion $\id(\top)\subseteq R$,
  where $\id(\top)$ denotes the identity role, this logic is a
  fragment of $\mathcal{ALCHQ}(\id)$~\cite{BaaderEA03}.
\end{expl}

\section{Conclusion}

\noindent We have formulated two local semantic conditions that
guarantee \PSPACE upper bounds for the satisfiability problem of modal
logics in a coalgebraic framework: the OSPMP (one-step polysize model
property) and its pointwise variant, the OSPPMP, which is weaker but
relies on additional assumptions on the coalgebraic semantics. Both
conditions allow a direct construction of shallow models and their
traversal in polynomial space. This complements earlier
work~\cite{SchroderPattinson06} where syntactic criteria have been
used --- in particular, both the OSPMP and the OSPPMP can be applied
even when no complete axiomatisation of the logic at hand is
known. Several instantiations of our results to logics studied in the
literature witness both their generality and their usefulness: Apart
from re-proving known \PSPACE upper bounds for the normal modal logics
$K$ and $T$ as well as for the conditional logics $\CK$, \CKId, and
\CKmp, we have
\begin{sparitemize}
\item given a systematic account of tight \PSPACE upper bounds in
  modal logics of quantitative uncertainty that establishes new
  complexity bounds in some cases;
\item obtained a new \PSPACE upper bound for Elgesem's modal logic of
  agency and for graded (even Presburger~\cite{DemriLugiez06}) modal
  logic over reflexive frames~\cite{Fine72}, and more generally for an
  extension of the description logic $\mathcal{ALCHQ}$ with reflexive
  roles~\cite{BaaderEA03};
\item established (to our knowledge: new) tight $\mi{NP}$ upper bounds for
  bounded-rank fragments of the conditional logics $\CK$, \CKId, and
  \CKmp.
\end{sparitemize}
\noindent Ongoing work focusses on the extension of our results to
\emph{iterative} modal logics, defined by frame conditions of higher
rank, which however --- in particular outside the realm of Kripke
semantics --- exhibit a tendency towards higher complexity or even
undecidability (indeed, it seems to be the case that all known
iterative \PSPACE-complete modal logics are normal). 

\bibliographystyle{abbrv} \bibliography{coalgml}

\begin{thebibliography}{10}

\bibitem{BaaderEA03}
F.~Baader, D.~Calvanese, D.~L. McGuinness, D.~Nardi, and P.~F. Patel-Schneider,
  editors.
\newblock {\em The Description Logic Handbook}.
\newblock Cambridge University Press, 2003.

\bibitem{Barr93}
M.~Barr.
\newblock Terminal coalgebras in well-founded set theory.
\newblock {\em Theoret.\ Comput.\ Sci.}, 114:299--315, 1993.

\bibitem{Chellas80}
B.~Chellas.
\newblock {\em Modal Logic}.
\newblock Cambridge Univ.\ Press, 1980.

\bibitem{CholvyEA05}
L.~Cholvy, C.~Garion, and C.~Saurel.
\newblock Ability in a multi-agent context: {A} model in the situation
  calculus.
\newblock In {\em Computational Logic in Multi-Agent Systems, CLIMA 2005},
  volume 3900 of {\em LNCS}, pages 23--36. Springer, 2006.

\bibitem{CirsteaPattinson04}
C.~C{\^i}rstea and D.~Pattinson.
\newblock Modular construction of modal logics.
\newblock In {\em Concurrency Theory}, volume 3170 of {\em LNCS}, pages
  258--275. Springer, 2004.

\bibitem{DemriLugiez06}
S.~Demri and D.~Lugiez.
\newblock {P}resburger modal logic is only {\PSPACE}-complete.
\newblock In {\em Automated Reasoning, IJCAR 06}, volume 4130 of {\em LNAI},
  pages 541--556. Springer, 2006.

\bibitem{Elgesem97}
D.~Elgesem.
\newblock The modal logic of agency.
\newblock {\em Nordic J.\ Philos.\ Logic}, 2:1--46, 1997.

\bibitem{FaginHalpern94}
R.~Fagin and J.~Y. Halpern.
\newblock Reasoning about knowledge and probability.
\newblock {\em J.\ ACM}, 41:340--367, 1994.

\bibitem{FaginHalpernMegiddo90}
R.~Fagin, J.~Y. Halpern, and N.~Megiddo.
\newblock A logic for reasoning about probabilities.
\newblock {\em Inform.\ Comput.}, 87:78--128, 1990.

\bibitem{FattorosiBarnabaDeCaro85}
M.~Fattorosi-Barnaba and F.~{De Caro}.
\newblock Graded modalities {I}.
\newblock {\em Stud.\ Log.}, 44:197--221, 1985.

\bibitem{Fine72}
K.~Fine.
\newblock In so many possible worlds.
\newblock {\em Notre Dame J.\ Formal Logic}, 13:516--520, 1972.

\bibitem{GovernatoriRotolo05}
G.~Governatori and A.~Rotolo.
\newblock On the axiomatisation of {E}lgesem's logic of agency and ability.
\newblock {\em J.\ Philos.\ Logic}, 34:403--431, 2005.

\bibitem{Halpern95}
J.~Halpern.
\newblock The effect of bounding the number of primitive propositions and the
  depth of nesting on the complexity of modal logic.
\newblock {\em Artificial Intelligence}, 75:361--372, 1995.

\bibitem{HalpernPucella02}
J.~Halpern and R.~Pucella.
\newblock A logic for reasoning about upper probabilities.
\newblock {\em J.\ Artificial Intelligence Res.}, 17:57--81, 2002.

\bibitem{HalpernPucella02b}
J.~Halpern and R.~Pucella.
\newblock Reasoning about expectation.
\newblock In {\em Uncertainty in Artificial Intelligence, UAI 02}, pages
  207--215. Morgan Kaufman, 2002.

\bibitem{HalpernMoses92}
J.~Y. Halpern and Y.~O. Moses.
\newblock A guide to completeness and complexity for modal logics of knowledge
  and belief.
\newblock {\em Artificial Intelligence}, 54:319--379, 1992.

\bibitem{Jacobs00}
B.~Jacobs.
\newblock Towards a duality result in the modal logic of coalgebras.
\newblock In {\em Coalgebraic Methods in Computer Science, CMCS 2000},
  volume~33 of {\em ENTCS}. Elsevier, 2000.

\bibitem{JonesParent07}
A.~Jones and X.~Parent.
\newblock Conventional signalling acts and conversation.
\newblock In {\em Advances in Agent Communication}, volume 2922 of {\em LNAI},
  pages 1--17. Springer, 2004.

\bibitem{Kurz01}
A.~Kurz.
\newblock Specifying coalgebras with modal logic.
\newblock {\em Theoret.\ Comput.\ Sci.}, 260:119--138, 2001.

\bibitem{Ladner77}
R.~Ladner.
\newblock The computational complexity of provability in systems of modal
  propositional logic.
\newblock {\em SIAM J.\ Comput.}, 6:467--480, 1977.

\bibitem{Lewis74}
D.~Lewis.
\newblock Intensional logics without iterative axioms.
\newblock {\em J.\ Philos.\ Logic}, 3:457--466, 1975.

\bibitem{OlivettiEA07}
N.~Olivetti, G.~L. Pozzato, and C.~Schwind.
\newblock A sequent calculus and a theorem prover for standard conditional
  logics.
\newblock {\em ACM Trans.\ Comput.\ Logic}.
\newblock To appear.

\bibitem{PacuitSalame04}
E.~Pacuit and S.~Salame.
\newblock Majority logic.
\newblock In {\em Principles of Knowledge Representation and Reasoning, KR 04},
  pages 598--605. AAAI Press, 2004.

\bibitem{Papadimitriou81}
C.~H. Papadimitriou.
\newblock On the complexity of integer programming.
\newblock {\em J.\ ACM}, 28:765--768, 1981.

\bibitem{Pattinson03}
D.~Pattinson.
\newblock Coalgebraic modal logic: Soundness, completeness and decidability of
  local consequence.
\newblock {\em Theoret.\ Comput.\ Sci.}, 309:177--193, 2003.

\bibitem{Pauly02}
M.~Pauly.
\newblock A modal logic for coalitional power in games.
\newblock {\em J.\ Logic Comput.}, 12:149--166, 2002.

\bibitem{Schroder05}
L.~Schr{\"o}der.
\newblock Expressivity of coalgebraic modal logic: the limits and beyond.
\newblock {\em Theoret.\ Comput.\ Sci.}
\newblock In press.

\bibitem{Schroder06}
L.~Schr{\"o}der.
\newblock A finite model construction for coalgebraic modal logic.
\newblock {\em J.\ Log.\ Algebr.\ Prog.}, 73:97--110, 2007.

\bibitem{SchroderPattinson06}
L.~Schr{\"{o}}der and D.~Pattinson.
\newblock {\PSPACE} reasoning for rank-1 modal logics.
\newblock In {\em Logic in Computer Science, LICS 06}, pages 231--240. IEEE,
  2006.
\newblock Extended version to appear in \emph{ACM Trans.\ Comput.\ Log.}

\bibitem{Tobies01}
S.~Tobies.
\newblock {\PSPACE} reasoning for graded modal logics.
\newblock {\em J.\ Logic Comput.}, 11:85--106, 2001.

\bibitem{Vardi89}
M.~Vardi.
\newblock On the complexity of epistemic reasoning.
\newblock In {\em Logic in Computer Science, LICS 89}, pages 243--251. IEEE,
  1989.

\end{thebibliography}
  
\end{document}